\documentclass[10pt]{article}
\usepackage{amsfonts}
\usepackage{amsmath}
\usepackage{amssymb}
\usepackage{color}
\newcommand{\F}{\mathbb{F}}
\newcommand{\wt}{\mbox{wt}}
\newcommand{\gen}{\mbox{gen}}
\newcommand{\Aut}{\mbox{Aut}}
\newtheorem{theorem}{Theorem}[section]
\newtheorem{lemma}[theorem]{Lemma}
\newtheorem{proposition}[theorem]{Proposition}
\newtheorem{corollary}[theorem]{Corollary}

\newenvironment{proof}{\noindent {\em Proof.}}{\hspace*{\fill} $\Box $\newline}

\begin{document}

\newcommand{\comment}[1]{} 

\title{ Classification of Binary Self-Dual [48,24,10] Codes with an Automorphism of Odd Prime Order}

\author{
Stefka Bouyuklieva \\
Faculty of Mathematics and Informatics \\
 Veliko Tarnovo University \\
5000 Veliko Tarnovo, Bulgaria \\
{Email: \tt stefka@uni-vt.bg}\\
\\
Nikolay Yankov \\
Faculty of Mathematics and Informatics \\
Shumen University \\
 9700 Shumen, Bulgaria \\
{Email: \tt jankov$\_$niki@yahoo.com} \\
and \\
Jon-Lark Kim \\
 Department of Mathematics \\
 University of Louisville \\
  Louisville, KY 40292, USA, and \\
Department of Mathematics, POSTECH \\
Pohang, Gyungbuk 790-784, South Korea \\
{Email: \tt jl.kim@louisville.edu} \\
}

\date{Jan. 19, 2012}

\maketitle

\begin{abstract}
The purpose of this paper is to complete the classification of
 binary self-dual $[48,24,10]$ codes with an automorphism of odd prime order.  We prove that if there is a self-dual $[48, 24, 10]$ code with an automorphism  of type $p$-$(c,f)$ with $p$ being an odd prime, then $p=3, c=16, f=0 $. By considering only an automorphism of type $3$-$(16,0)$, we prove that there are exactly $264$ inequivalent self-dual $[48, 24, 10]$ codes with an automorphism of odd prime order, equivalently, there are exactly $264$ inequivalent cubic self-dual $[48, 24, 10]$ codes.
\end{abstract}



\section{  Introduction}

A linear $[n,k]$ code over $\mathbb F_q$ is a $k-$dimensional subspace of $\mathbb F_q^n$. A linear code $C$ is called {\em self-dual} if it is equal to its dual $C^{\perp}=\{x \in \mathbb F_q^n~|~ x \cdot c = 0 {\mbox{ for any }} c \in C\}$. The classification of binary self-dual codes was initiated and done up to lengths $20$ by V. Pless~\cite{PlessSO}. Then the classification of self-dual codes has been one of the most active research topics (see \cite{RS}, \cite{Huf2005}). The classification of binary self-dual $[38, 19, 8]$ codes has been recently done by Aguilar-Melchor, Gaborit, Kim, Sok, and Sol\'{e}~\cite{AGKSS} and independently by Betsumiya, Harada and Munemasa \cite{BetHarMun}. Very recently,
Bouyuklieva and Bouyukliev~\cite{BouBou} have classified all binary self-dual $[38, 19]$ codes.

In this paper, we are interested in the classification of binary
self-dual $[48, 24, 10]$ codes with an automorphism of odd prime
order. It was motivated by the following reasons. Bonnecaze, et.
al.~\cite{BonBraDouNocSol} constructed binary self-dual codes with a
fixed free automorphism of order $3$, called {\em cubic self-dual
codes} due to the correspondence with self-dual codes over a ring
$\mathbb F_2[Y]/(Y^3-1)$. They gave a partial list of binary cubic
self-dual codes of lengths $\le 72$ by combining binary self-dual
codes and Hermitian self-dual codes. Later, Han, et.
al.~\cite{HanKimLeeLee} have given the classification of binary
cubic optimal self-dual codes of length $6k$ where $k=1,2, \dots,7$.
Hence it is natural to ask exactly how many binary cubic self-dual
optimal $[48, 24, 10]$ codes exist and we answer it in this paper.
On the other hand, we have noticed that Huffman~\cite[Table
2]{Huf2005} listed all possible values of the type $p$-$(c,f)$ with
$p$ odd for an automorphism of a self-dual $[48,24,10]$ code. They
are $11$-$(4,4)$, $7$-$(6,6)$, $5$-$(8,8)$, $3$-$(14, 6)$, and
$3$-$(16,0)$. We will show that the first four types are not
possible. Therefore, the classification of binary self-dual $[48,
24, 10]$ codes with the last type $3$-$(16,0)$ coincides with that
of binary cubic self-dual codes.

\medskip

There are two possible weight enumerators for self-dual
$[48, 24, 10]$ codes \cite{C-S}:
\begin{eqnarray}
  W_{48,1}(y) &=&1 + 704y^{10} + 8976y^{12} +56896y^{14} + 267575y^{16}+\dots\label{one}\\
  W_{48,2}(y) &=&1 + 768y^{10} + 8592y^{12} +57600y^{14} + 267831y^{16} +\dots \label{two}.
\end{eqnarray}
     Brualdi and Pless \cite{Brualdi-Pless} found a [48,24,10] code with weight enumerator $W_{48,1}$.
     The order of its group of automorphisms is 4. A code with weight enumerator $W_{48,2}$ is given in \cite{C-S}.

    The first author~\cite{Bou97} showed that any code with $W_{48,1}(y)$  has no automorphism of odd prime order and that any code with $W_{48,2}(y)$ has a group of automorphisms of order $2^l 3^s$ for some integers $l\ge 0$ and $s\ge 0$. However, this result has received less attention and hence we will include it briefly. In this paper, we prove that if there is a self-dual $[48, 24, 10]$ code with an automorphism  of type $p$-$(c,f)$ with $p$ being an odd prime, then $p=3, c=16, f=0 $. Therefore by considering only an automorphism of type $3$-$(16,0)$, we prove that there are exactly $264$ inequivalent self-dual $[48, 24, 10]$ codes with an automorphism of odd prime order. To do that we apply the method for constructing binary self-dual codes
possessing an automorphism of odd prime order (see \cite{Huff},
\cite{Yorgov56}, \cite{Yorus}).

\section{Construction method}

Let $C$ be a binary self-dual code of length $n$ with an
automorphism $\sigma $ of prime order $p\geq 3$ with exactly $c$
independent $p$-cycles and $f=n-cp$ fixed points in its
decomposition. We may assume that
\begin{equation}\label{sigma}
\sigma=(1,2,\cdots,p)(p+1,p+2,\cdots,2p)\cdots((c-1)p+1,(c-1)p+2,\cdots,cp),
\end{equation}
and shortly say that $\sigma$ is of \emph{type} $p$-$(c,f)$.

 We begin with a theorem which gives a useful restriction for the type of the
automorphism.

 \begin{theorem}{\rm(\cite{Yorus})}\label{aut_type}
  Let $C$ be a binary self-dual $[n,n/2,d]$ code with
 an automorphism of type $p$-$(c,f)$ where $p$ is an odd prime. Denote
    $g(k)=d+\lceil \frac{d}{2}\rceil +\cdots +\lceil \frac{d}{2^{k-1}}\rceil$.
     Then:

       (i) $pc\ge g(\frac{(p-1)c}{2})$ and if $d\le 2^{(p-1)c/2 -2}$
     the equality does not occur;

      (ii) if $f>c$ then $f\ge g(\frac{f-c}{2} )$ and if
     $d\le 2^{(f-c)/2 -2}$ the equality does not occur;

     (iii) if $2$ is a primitive root modulo $p$ then $c$ is even.
\end{theorem}

Applying the theorem for the parameters $n=48$ and $d=10$, we
obtain

\begin{corollary}\label{cor:types1}
Any putative automorphism of an odd prime order for a singly-even
self-dual $[48,24,10]$ code is of type $47$-$(1,1)$, $23$-$(2,2)$,
$11$-$(4,4)$, $7$-$(6,6)$, $5$-$(8,8)$, $3$-$(12,12)$,
$3$-$(14,6)$, or $3$-$(16,0)$.
\end{corollary}

Denote the cycles of $\sigma $ by
    $\Omega_1 =\{ 1,2,\ldots,p\} $, $\Omega _2,\ldots,\Omega _c $, and the fixed points by
    $\Omega_{c+1} =\{ cp+1\} ,\ldots,\Omega_{c+f} =\{ cp+f=n\} $. Define
\begin{eqnarray*}
F_\sigma(C)&=&\{ v \in C\mid \sigma(v)=v\},\\
E_\sigma(C)&=&\{ v\in C\mid\wt(v\vert \Omega_i)\equiv 0\pmod 2, \
i=1,\cdots,c+f\},
\end{eqnarray*}
where $v\vert\Omega_i $ is the restriction of $v$ on $\Omega_i $.

\begin{theorem}\label{directsum}
{\rm(\cite{Huff})} $C=F_\sigma(C)\oplus E_\sigma(C)$,
$\dim(F_\sigma)=\frac{c+f}{2}$,
$\dim(E_\sigma)=\frac{c(p-1)}{2}$.
\end{theorem}

We have that $v\in F_\sigma (C)$ if and only if $v\in C$ and $v$
is constant on each cycle. The cyclic group generated by
$\sigma$ splits the set of codewords into disjoint orbits which
consists of $p$ or 1 codewords. Moreover, a codeword $v$ is the
only element in a orbit only if $v\in F_\sigma (C)$. Using that
all codewords in one orbit have the same weight, we obtain the
following proposition.

\begin{proposition}\label{weight_distrib}
Let $(A_0,A_1,\dots,A_n)$ and $(B_0,B_1,\dots,B_n)$ be the weight
distributions of the codes $C$ and $F_{\sigma}(C)$, respectively.
Then $A_i\equiv B_i\pmod p$.
\end{proposition}

Proposition \ref{weight_distrib} eliminates the first two
types from Corollary \ref{cor:types1}. In fact, these cases were eliminated by Huffman~\cite[Appendix]{Huf2005} in a different way.

\begin{corollary}\label{cor:types2}
If $C$ is a self-dual $[48,24,10]$ code, then $C$ does not have
automorphisms of orders $47$ and $23$.
\end{corollary}

\begin{proof}
Let $\sigma$ be an automorphism of $C$. If $\sigma$ is of type
$47$-$(1,1)$ then $F_{\sigma}(C)$ is the repetition [48,1,48] code
and therefore $B_{10}=0$. Since neither 704 nor 768 is congruent
0 modulo 47, this case is not possible.

If the type is $23$-$(2,2)$ then $B_i=0$ for $0<i<24$. Therefore
$B_{10}=0$ and $A_{10}\not\equiv B_{10}\pmod{23}$ - a
contradiction.
\end{proof}

To understand the structure of a self-dual code $C$ invariant
under the permutation (\ref{sigma}), we define two maps.  The
first one is the projection map $\pi :F_\sigma (C)\to \F_{2}^{c+f}
$ where $(\pi(v))_i =v_j $ for some $j\in\Omega_i$,
$i=1,2,\ldots,c+f $, $v\in F_\sigma (C)$.

Denote by $E_\sigma(C)^{*}$ the code $E_\sigma(C)$ with the last
$f$ coordinates deleted. So $E_\sigma(C)^{*}$ is a self-orthogonal
binary code of length $pc$. For $v$ in $E_\sigma (C)^*$ we let
$v\vert\Omega_i=(v_0,v_1,\cdots,v_{p-1})$ correspond to the
polynomial $v_0+v_1 x+\cdots+v_{p-1}x^{p-1}$ from ${\cal P}$,
where ${\cal P}$ is the set of even-weight polynomials in $\F_2
[x]/(x^p-1)$. ${\cal P}$ is a cyclic code of length $p$ with
generator polynomial $x-1$. Moreover, if 2 is a primitive root
modulo $p$, ${\cal P}$ is a finite field with $2^{p-1}$ elements
\cite{Huff}. In this way we obtain the map
$\varphi:E_\sigma(C)^{*}\to {\cal P}^c $. The following theorems
give necessary and sufficient conditions a binary code with an
automorphism of type (\ref{sigma}) to be self-dual.

\begin{theorem}{\rm (\cite{Yorus})}\label{thm2}
A binary $[n,n/2]$ code $C$ with an automorphism $\sigma$ is
self-dual if and only if the following two conditions
hold:

$(i)$ $C_{\pi}=\pi(F_\sigma(C))$ is a binary self-dual code of
length $c+f$,

$(ii)$ for every two vectors $u, v\in
C_{\varphi}=\varphi(E_{\sigma}(C)^{*})$ we have
$\sum\limits_{i=1}^{c} u_i(x)v_i(x^{-1})=0$.
\end{theorem}

\begin{theorem}{\rm (\cite{Huff})}
\label{primroot} Let $2$ be a primitive root modulo $p$. Then the
binary code $C$ with an automorphism $\sigma$ is self-dual if and only if the
following two conditions hold:

 $(i)$ $C_\pi$ is a self-dual binary code of length $c + f$;

 $(ii)$ $C_\varphi$ is a self-dual code of length $c$ over the field ${\cal P}$ under the inner product
 $(u,v)=\sum\limits_{i=1}^{c}{u_iv_i^{2^{(p-1)/2}}}.$
\end{theorem}

To classify the codes, we need additional conditions for
equivalence. That's why we use the following theorem:

\begin{theorem} {\rm (\cite{Yorgov56})}\label{thm:eq}
The following transformations preserve the decomposition and send
the code $C$  to an equivalent one:

a) the substitution $x\to x^t$ in $C_{\varphi}$, where $t$ is an
integer, $1\le t\le p - 1$;

b) multiplication of the $j$th coordinate of $C_{\varphi}$ by
$x^{t_j}$ where $t_j$ is an integer, $0\le t_j\le p - 1$, $j =
1,2,\dots,c$;

c) permutation of the first $c$ cycles of $C$;

d) permutation of the last $f$ coordinates of $C$.
   \end{theorem}

\section{Codes with an Automorphism of Odd Prime Order}

\subsection{  Codes with an Automorphism of Order 3}

In this section we first classify self-dual codes with an automorphism of order 3.

Let $C$ be a self-dual $[48,24,10]$ code with an automorphism of
order 3. According to Corollary \ref{cor:types1} this automorphism
is of type 3-(12,12), 3-(14,6) or 3-(16,0). In this section we
prove that only the type 3-(16,0) is possible. Moreover, we
classify all binary self-dual codes with the given parameters
which are invariant under a fixed point free permutation of order
3.

\begin{proposition}\label{prop:p3c12}
Self-dual $[48,24,10]$ codes with an automorphism of type
$3$-$(12,12)$ do not exist.
\end{proposition}

\begin{proof}
 In this case the code $C_\varphi$ is a self-dual $[12,6]$ code over the field
     $\mathcal{P}=\{ 0,e(x)=x+x^2,xe(x),x^2 e(x)\} $ under the inner product
     $(u,v)=u_1v_1^2 + u_2 v_2^2 +\cdots +u_{12}v_{12}^2$.
    The highest possible minimum distance of a quaternary Hermitian self-dual $[12,6]$ code is 4 (see \cite{CPS}), hence the minimum distance of
    $E_\sigma (C)$ can be at most 8 - a conflict with the minimum distance
    of C.\end{proof}

   \begin{proposition}\label{prop:p3c14}
Self-dual $[48,24,10]$ codes with an automorphism of type
$3$-$(14,6)$ do not exist.
\end{proposition}

\begin{proof} Assume that $\sigma =(1,2,3)(4,5,6)\ldots
(40,41,42)$ is an automorphism of the self-dual $[48,24,10]$ code
$C$. Then $C_\pi$ is a binary self-dual [20,10,4] code. There are
exactly 7 inequivalent self-dual [20,10,4] codes, namely $J_{20}$,
$A_8\oplus B_{12}$, $K_{20}$, $L_{20}$, $S_{20}$, $R_{20}$ and
$M_{20}$ (see \cite{PlessSO}). If $C_\pi$ is equivalent to any of
these codes there is a vector $v=(v_1,v_2)$ in $C_\pi$ with
$v_1\in \F_{2}^{14}$, $v_2\in \F_{2}^6$ and $wt(v_1)=wt(v_2)=2$.
Thus the vector $\pi ^{-1}(v)\in F_\sigma (C)$ has weight 8 which
contradicts the minimum distance.\end{proof}

Let now $C$ be a singly-even  $[48,24,10]$ code, possessing an
automorphism with 16 cycles of length 3 and no fixed point in its
decomposition into independent cycles.


According to Theorem \ref{thm2} the subcode $C_{\pi}$ is a
singly-even binary self-dual $[16,8,\geq 4]$ code. There is one
such code, denoted by $2d_8$ in \cite{PlessSO}, with a generator
matrix
$$G_B=\left(\begin{array}{c}
1000001110010010\\
0100001110011101\\
0010001110001111\\
0001001100011010\\
0000101010011010\\
0000010110011010\\
0000000001010110\\
0000000000110011
\end{array}\right).$$
It is more convenient for us to denote this code by $B$. The
automorphism group of $B$ is generated by the permutations
$(1,12,4,2,13,15,9,3)(5,16,8,11)(6,14)(7,10)$ and
$(1,8,7)(5,6,13)(10,12)(14,15)$. Its order is 76728.


According to Theorem \ref{thm2} the subcode $C_{\varphi}$ is a
quaternary Hermitian self-dual $[16, 8,\geq 5]$ code. There are
exactly 4 inequivalent such codes $1_{16}$, $1_6+2f_5$, $4f_4$,
and $2f_8$ \cite{CPS}. Their generator matrices in standard form
are $G_i=(I|X_i)$, $i=1,\dots,4$, where
\[
X_1=\left(\begin{array}{cccccccc}
0&0&\omega&1&1&1&0&\omega^2\\
\omega^2&\omega^2&0&\omega^2&\omega&\omega^2&1&1\\
\omega&0&1&\omega&0&1&\omega&0\\
1&\omega^2&\omega^2&\omega&\omega^2&\omega^2&\omega^2&0\\
0&1&1&1&\omega^2&\omega^2&\omega&1\\
\omega^2&1&\omega^2&0&\omega&0&\omega&0\\
\omega&0&\omega&\omega^2&1&\omega&1&\omega^2\\
0&1&\omega^2&\omega^2&\omega^2&\omega^2&\omega^2&\omega^2\\
\end{array}\right), \
X_2=\left(\begin{array}{cccccccc}
\omega^2&1&\omega^2&0&\omega^2&0&\omega&0\\
1&\omega^2&\omega&\omega^2&1&1&0&1\\
\omega&1&1&\omega^2&1&\omega^2&\omega^2&0\\
0&\omega&0&\omega^2&0&1&\omega^2&\omega^2\\
1&\omega&0&\omega&1&1&\omega^2&1\\
0&\omega^2&\omega&1&\omega^2&\omega&\omega&\omega^2\\
\omega&0&0&1&\omega^2&\omega^2&\omega^2&0\\
0&\omega&\omega^2&1&\omega^2&\omega^2&\omega^2&\omega^2\\
\end{array}\right),
\]
\[
X_3=\left(\begin{array}{cccccccc}
\omega&0&\omega^2&0&0&\omega^2&1&\omega^2\\
1&\omega^2&\omega&\omega^2&1&1&0&1\\
\omega^2&\omega^2&\omega^2&\omega&0&\omega&\omega&1\\
1&\omega^2&0&\omega^2&\omega^2&\omega&0&0\\
1&\omega&0&\omega&1&1&\omega^2&1\\
0&1&\omega^2&\omega&\omega&\omega^2&\omega^2&\omega\\
1&1&1&\omega&0&0&0&\omega^2\\
1&\omega&0&\omega^2&\omega^2&\omega^2&\omega^2&\omega^2\\
\end{array}\right), \
X_4=\left(\begin{array}{cccccccc}
1&0&\omega^2&\omega^2&\omega&1&\omega^2&1\\
1&\omega^2&\omega&\omega^2&1&1&0&1\\
\omega&1&1&\omega^2&1&\omega^2&\omega^2&0\\
\omega&\omega^2&0&0&1&0&\omega&\omega\\
1&\omega&0&\omega&1&1&\omega^2&1\\
1&0&\omega&\omega^2&0&1&1&0\\
\omega&0&0&1&\omega^2&\omega^2&\omega^2&0\\
0&\omega^2&1&\omega&1&1&1&1\\
\end{array}\right).\]

Denote by $C_i^\tau$ the self-dual $[48,24,10]$ code with a
generator matrix $$\gen\ C_i^\tau=\left(
  \begin{array}{c}
\pi^{-1}(\tau B)\\
\varphi^{-1}(X_{i})\\
  \end{array}
  \right),$$
where $\tau$ is a permutation from the symmetric group $S_{16}$,
and  $1\le i\le 4$. We use the following.

\begin{lemma}{\rm(\cite{Yorus})} If $\tau_1$  and  $\tau_2$ are in one and the same right coset of $\Aut(B)$ in $S_{16}$, then
$C_i^{\tau_{1}}$  and  $C_i^{\tau_{2}}$ are equivalent.
\end{lemma}

In order to classify all codes we have considered all
representatives of the right transversal of  $S_{16}$ with respect
to $\Aut(B)$. The obtained inequivalent codes and the orders of
their automorphism groups are listed in Tables \ref{Table:1}- \ref{Table:4} at the end.

\begin{proposition}\label{prop:3-16-0} There are exactly $264$ inequivalent binary $[48, 24, 10]$ self-dual codes with an automorphism of type
$3$-$(16,0)$.
\end{proposition}

\begin{corollary}
There are exactly $264$ inequivalent binary cubic self-dual $[48, 24, 10]$ codes.
\end{corollary}

\subsection{Codes with an Automorphism of Orders 5, 7, and 11}
\label{sec-11}

The first author~\cite{Bou97} showed that there are no self-dual $[48, 24, 10]$  codes with an automorphism of orders 5, 7, and 11. However, these results have received less attention since even Huffman in his survey paper~\cite{Huf2005} could not eliminate these types of automorphisms. Hence it is worth sketching the nonexistence of self-dual $[48, 24, 10]$ codes with an automorphism of orders 5, 7, and 11.

\medskip

        Let $C$ be a binary singly-even self-dual $[48,24,10]$ code with an automorphism of order 11 and
   $$\sigma =(1,2,\ldots ,11)(12,\ldots ,22)(23,\ldots ,33)(34,\ldots ,44)$$
    be an automorphism of $C$. Then $\pi (F_\sigma (C))$ is a binary self-dual
 code of length 8.

\begin{lemma}
The code $\pi (F_\sigma (C))$ is generated by the matrix
$(I_4\vert I_4 +J_4)$
    up to a permutation of the last four coordinates. Here $I_4$ is the
identity
    matrix and $J_4$ is the all-one matrix.
    \end{lemma}

    Since 10 is the multiplicative order of 2 modulo 11, $C_{\varphi}=\varphi (E_\sigma
(C)^{*})$ is a self-dual [4,2] code over the field ${\cal P}$ of
even-weight polynomials in $F_2 [x]/(x^{11}-1)$ with $2^{10}$
elements under the inner product
\begin{equation}\label{product11}
   (u,v)=u_1 v_{1}^{32} +u_2 v_{2}^{32}+u_3 v_{3}^{32} +u_4 v_{4}^{32}.
   \end{equation}

      \begin{lemma}\label{lemma:phi11}
        $C_{\varphi}$ is a $[4,2,3]$ self-dual code over the field ${\cal P}$.
    \end{lemma}

By considering all possibilities of $C_{\varphi}$, one can get the following.

        \begin{theorem}{\rm (\cite{Bou97})}
        There does not exist a self-dual $[48,24,10]$ code with an automorphism
    of order $11$.
    \end{theorem}

 Let $C$ be a binary singly-even self-dual $[48,24,10]$ code with
    an automorphism of order 7 and let
    $$\sigma =(1,2,\ldots ,7)(8,\ldots ,14)\cdots (36,\ldots ,42).$$
    Now $\pi (F_\sigma (C))$ is a [12,6] self-dual code. Since the minimal
    distance of $F_\sigma (C)$  is at least 10, $\pi (F_\sigma (C))$
    has a generator matrix of the form $(I_6\vert A)$, where $I_6$ is
    the identity matrix. We obtain a unique possibility for $A$ up to a
permutation
    of its columns:
\begin{equation}\label{eq:p=7pi}
    \left(
    \begin{array}{cccccc}
1 & 1 & 1 & 0 & 0 & 0\\ 1 & 1 & 0&1& 0 & 0\\ 1 &1& 0&0&1&0\\
1&1&0&0&0&1\\ 0&1&1&1&1&1\\ 1&0&1 & 1 & 1 & 1\\
\end{array}\right).
\end{equation}


Since $2^3\equiv 1\pmod 7$, 2 is not a primitive root modulo 7 and
${\cal P}$ is not a field. Now ${\cal P}=I_1\oplus I_2$, where
$I_1$ and $I_2$ are cyclic codes generated by the idempotents
$e_1(x)=1+x+x^2+x^4$ and $e_2(x)=1+x^3+x^5+x^6$, respectively, so
$$I_j=\{ 0,e_j(x),xe_j(x),\ldots ,x^6 e_j(x)\} ,j=1,2.$$ Moreover, $I_1$ and $I_2$ are fields of 8 elements \cite{Yorus}.

In this case $C_\varphi=\varphi (E_\sigma (C)^{*})=M_1\oplus M_2$, where
$M_j=\{u\in C_{\varphi}\mid u_i\in I_j, i=1,\dots,6\}$ is a linear
code over the field $I_j$, $j=1,2$, and
$\dim_{I_1}M_1+\dim_{I_2}M_2=6$ \cite{Yorus}. Since the minimum
weight of the code $C$ is 10, every vector of $C_{\varphi}$ must
contain at least three nonzero coordinates. Hence the minimum
weight of $M_j$ is at least 3, and so the dimension of $M_j$ is at
least 2, $j=1,2$.

By Theorem \ref{thm2} for every two vectors
$(u_1(x),\ldots,u_6(x))$ from $ M_1$ and $(v_1(x),\ldots ,v_6(x))$
from $M_2$ we have
$$u_1(x) v_1(x^{-1})+\cdots +u_6(x)v_6(x^{-1})=0.$$
Since $e_1(x^{-1})=e_2(x)$ and $e_1(x)e_2(x)=0$, $M_2$ determines
the whole code $C_{\varphi}$. The substitution $x\to x^3$ in
$\varphi (E_\sigma (C)^{*})$ interchanges $M_1$ and $M_2$ and
therefore we may assume that $dim_{I_1}M_1\ge dim_{I_2}M_2$. We
have two cases, $dim_{I_2}M_2=2$ and $dim_{I_2}M_2=3$. Each case does not produce a self-dual $[48,24,10]$ code with an automorphism
    of order $7$ as follows.

        \begin{theorem}{\rm (\cite{Bou97})}
        There does not exist a self-dual $[48,24,10]$ code with an automorphism
    of order $7$.
    \end{theorem}


According to Corollary \ref{cor:types1}, if the self-dual
$[48,24,10]$ code $C$ has an automorphism $\sigma$ of order 5 then
$\sigma$ is of type $5$-$(8,8)$. Here we prove that this is not
possible.

 Let $C$ has an automorphism $\sigma$ of type
  $5$-$(8,8)$. Then $C_\varphi$ is a self-dual $[8,4]$ code over the field ${\cal P}$ with 16 elements under the inner product
    $(u,v)=u_1 v_{1}^4+u_2 v_{2}^4+\cdots +u_8 v_{8}^4$, $u,v\in C_{\varphi}$. There is
    one-to-one correspondence between the elements of the field ${\cal
    P}$ and the set of $5\times 5$ cirulants with rows of even
    weight defined by
    \[ a_0+a_1x+a_2x^2+a_3x^3+a_4x^4 \ \mapsto \ \left(
    \begin{array}{ccccc}
    a_0&a_1&a_2&a_3&a_4\\
    a_4&a_0&a_1&a_2&a_3\\
    a_3&a_4&a_0&a_1&a_2\\
    a_2&a_3&a_4&a_0&a_1\\
    a_1&a_2&a_3&a_4&a_0\\
    \end{array}\right).
    \]
    Moreover, the rank of a nonzero circulant of this type is 4.
    Therefore, to any nonzero vector $u\in C_{\varphi}$ correspond a
    subcode of $E_{\sigma}(C)^*$ of length 40 and dimension 4.
    Moreover, the effective length of this subcode is $5\wt(u)$.
    Since self-orthogonal $[15,4,10]$ and $[20,4,10]$ codes do not
    exist (see \cite{BBGO}), we have $\wt(u)\ge 5$. Hence $C_\varphi$ must be an MDS
    $[8,4,5]$ Hermitian self-dual code over the field ${\cal P}\cong
    GF(16)$. Huffman proved in \cite{Huff} that such codes do not exists.
    So a self-dual $[48,24,10]$ with an automorphism of type
    5-(8,8) does not exist.

 \begin{theorem}{\rm (\cite{Bou97})}\label{thm:aut5}
        A self-dual $[48,24,10]$ code $C$ does not have automorphisms of
    order 5.
    \end{theorem}

\begin{lemma}{\rm (\cite{Bou97})}\label{lem:order3}
A self-dual $[48, 24, 10]$ code with an automorphism of order $3$ has weight enumerator $W_{48, 2}(y)$.
\end{lemma}

Putting together the above results, we have the following theorems.

\begin{theorem}{\rm (\cite{Bou97})}
If $C$ is a singly-even self-dual $[48,24,10]$ code with weight enumerator $W_{48, 1}(y)$, the automorphism group of $C$ is of order $2^s$ with $s \ge 0$.
\end{theorem}

\begin{theorem}{\rm (\cite{Bou97})}
If $C$ is a self-dual singly-even $[48,24,10]$ code with weight
enumerator $W_{48, 2}(y)$, then the automorphism group of $C $is of order $2^s 3^t$ with $s \ge 0, t \ge 0$.
\end{theorem}

Therefore, using Proposition~\ref{prop:3-16-0}, we summarize our
result below.

\begin{theorem}
If there is a self-dual $[48, 24, 10]$ code with an automorphism  of
type $p$-$(c,f)$ with $p$ being an odd prime, then $p=3, c=16, f=0
$. Moreover, there are exactly $264$ inequivalent binary $[48, 24,
10]$ self-dual codes with an automorphism of odd prime order, which
is in fact of type $3$-$(16,0)$. Hence there are exactly $264$
inequivalent binary cubic self-dual $[48, 24, 10]$ codes.
\end{theorem}

\section*{Acknowledgment}
S. Bouyuklieva was partially supported by VTU Science Fund under Contract RD-642-01/26.07.2010.
J.-L. Kim was partially supported by the Project Completion Grant (year 2011-2012) at the University of Louisville.


\begin{table}[htb]
\centering \caption{Generating permutations and $|\Aut(C)|$ for
codes with $C_{\varphi}= 1_{16}$} \vspace*{0.2in} \label{Table:1}
\small
\begin{tabular}{|l|c||l|c||l|c|}
\hline
\text{permutation}& $|\Aut|$&   \text{permutation}& $|\Aut|$&   \text{permutation}& $|\Aut|$\\
\hline
(5,6)(12,14)&3&(2,3,14,8,12,6)(5,9,11)&6&(2,14,8,12,6)(5,9,11)&6\\
\hline
(3,6,5,12,10,9,11)&6&(3,6,7,5,12,14,15,9,11)&6&&\\
\hline
\end{tabular}
\end{table}

\begin{table}[htb]
\centering \caption{Generating permutations and $|\Aut(C)|$ for
codes with $C_{\varphi}= 1_6+2f_5$} 
\label{Table:2} \small
\begin{tabular}{|l|c||l|c||l|c|}
\hline
\text{permutation}& $|\Aut|$&   \text{permutation}& $|\Aut|$&   \text{permutation}& $|\Aut|$\\
\hline
(1,2,16,14,15,10,13)&3&(1,7,12,13,5,2,10,11,9)&3&(1,12,16,9)(2,10,13,6,5)&3\\
\hline
(2,10,13,4,8,7,12,16,11,9)&3&(2,11,16,5,10,8,12,13)&3&(2,12)(4,6,8,7)(14,16)&3\\
\hline
(2,12,9,11,14,6,7,8)&3&(2,13)(3,7,12,6,11,10)&3&(2,14,12,13,4,6,11,9,8,3)&3\\
\hline
(2,15,3,11,10,13)&3&(2,8,11,14,16,13,4)(6,12)&3&(2,8,11,5,12,13)&3\\
\hline
(2,8,7,15,11,16)(6,14)&3&(2,9,7,12,10,16,13)(6,15)&3&(2,16,8,13)(6,11,7,12,10,9)&3\\
\hline
(3,10)(4,7,11,13)(6,12,8)&3&(2,3,11,10,6,5,8,4,12)&6&(2,8,14,6,5,9,15)&6\\
\hline
(2,8,15,9,11,10,6,5)&6&(2,12,3)(7,10,9,8)&6&&\\
\hline
\end{tabular}
\end{table}

\begin{table}[htb]
\centering \caption{Generating permutations and $|\Aut(C)|$ for
codes with $C_{\varphi}= 4f_4$} 
\label{Table:3} \small
\begin{tabular}{|l|c||l|c||l|c|}
\hline
\text{permutation}& $|\Aut|$&   \text{permutation}& $|\Aut|$&   \text{permutation}& $|\Aut|$\\
\hline
(1,2,10,5,6,11,13,16)(3,4)&3&(1,14,4,10,11,7,9,3)(8,15)&3&(1,16,7)(2,15,14,6)(4,12,5)&3\\
\hline
(1,16,8,6,14,4,11,12,7,3)&3&(2,3,6,9,8,12,7,16,10)&3&(2,3,9)(4,7,10,5)&3\\
\hline
(2,5,4,9,3,14,6)(12,16)&3&(2,5,6,15)(7,10,11)(9,12)&3&(2,6,10,3,9,15,8,12,4,16)&3\\
\hline
(2,6,16,10)(4,15,5)(7,12)&3&(2,6,16,10)(7,9,8,12)&3&(2,6,3,5,12,13)(7,9,10,16)&3\\
\hline
(2,6,5,12,9,8,14,11)&3&(2,7,9,8,12,13)(5,16,10,14)&3&(2,8,12)(4,16,9,15,6,10)&3\\
\hline
(2,8,12,13)(5,14,9,11,15,6)&3&(2,8,4,3,6,11)(5,15)&3&(2,8,7,15,9,11)(5,10)&3\\
\hline
(2,9,11,10,5,8,12,15,6)&3&(2,9,11,8,14,10,6,5,15)&3&(2,9,15,8,12,6,11,5)&3\\
\hline
(2,9,6,5,12,13)&3&(2,9,8,10,14,5,6)&3&(2,9,8,11,6)(5,16,10)&3\\
\hline
(2,9,8,12,3,6,11)&3&(2,9,8,16,14,6,5,12,3,13)&3&(2,9,8,6)(5,10)&3\\
\hline
(2,9,8,7,14,10,16)(6,12)&3&(2,9,8,7,4,12,15,13)&3&(2,10,5,12,9,8,11,15,6)&3\\
\hline
(2,10,5,16,7,6,4,12,15,13)&3&(2,10,5,4,11,15,9,12,6)&3&(2,11,5,4,3,14,6)(9,15)&3\\
\hline
(2,11,9,12)(5,15)(7,10)&3&(2,13)(3,5,4)(6,16,12,14)&3&(2,13)(5,11,12,9,8,14,6)&3\\
\hline
(2,13)(5,12,16,6,7,9,11,10)&3&(2,14,12,13)(3,5,4)(6,16)&3&(2,14,6)(3,11,5,4,15,9)&3\\
\hline
(2,14,7,5,4,15,8,10,11)&3&(2,14,9,12,6)(4,11,10,5)&3&(3,4)(6,15,11,8,7,14)&3\\
\hline
(3,5)(4,12,16,6)(7,15,14)&3&(3,5)(6,15)(7,12)(14,16)&3&(3,5)(6,15,7,12)&3\\
\hline
(3,5)(6,9,15,10,16,7,12)&3&(3,5,15)(6,10,12)(9,11)&3&(3,5,15,6,10)(9,11)&3\\
\hline
(3,5,15,8,16,6,10)(9,11)&3&(3,6,11,9,5,8,12,10)&3&(3,6,9,15,7,12,5)&3\\
\hline
(3,7,10,16,5,15,11,9,8)&3&(3,7,12,6,16,10,15,5)&3&(3,7,4,12,5)(6,15)(10,11)&3\\
\hline
(3,11)(5,15)(7,8,10,9,12)&3&(3,11,6,4,12,7,15,5)&3&(3,11,8,14,6,10,5,15)&3\\
\hline
(3,11,8,15)(5,14)(6,10)&3&(3,14,11,5,4,12,6,15,7)&3&(3,14,6,15,16,7,4,12,5)&3\\
\hline
(3,16)(5,15,9,12)(7,8,10)&3&(3,16,5,15,9,12,7,8,10)&3&(3,16,6,15,9,10)(5,12)&3\\
\hline
(3,16,7,12,9,5)(6,14)&3&(3,16,7,4,15,6,10,12,5)&3&(3,16,9,5,12)(6,10)(7,14)&3\\
\hline
(5,12,7,10,9,8,14)&3&(5,12,9,11,8,14)(6,10)&3&(5,14,16)(6,11,8,15)&3\\
\hline
(5,15,9,11,8,16,6,10)&3&(7,11,12,9,8)&3&(1,15,7,6,10,11,8,14,4,3)&6\\
\hline
(2,5,3,7)(6,16,14)&6&(2,6,12,5,15)(9,11)&6&(2,6,7,15)(5,12)(9,11)&6\\
\hline
(2,7,10,13)(3,9,8,12,14,5,16)&6&(2,7,10,5,16,9,8,12,13)&6&(2,7,15)(5,6,12,16,9,10)&6\\
\hline
(2,7,8,15,9,10,5,12,11)&6&(2,9,11)(5,8,15,14,6,12)&6&(2,9,11,10,5,16,6,12,13)&6\\
\hline
(2,9,11,6,16,5,8,12,13)&6&(2,9,6,5,11,12,13)(8,14)&6&(2,9,6,5,14,8,11,16,12,13)&6\\
\hline
(2,9,8,14,6,12)(7,15)&6&(2,9,8,14,6,15)(7,12)&6&(2,13)(5,12,9,11)(6,16,14)&6\\
\hline
(2,13)(5,16,7,10)(6,12,9)&6&(2,13)(5,6,12)(7,16)(9,10,11)&6&(2,13)(5,8,12,9,11)(6,16)&6\\
\hline
(2,13)(5,8,16,6,12,9,11)&6&(2,14)(5,12,9,8,15,6,11)&6&(2,14,6,15,9,11)(5,8,12)&6\\
\hline
(2,16,7,10,5,15,9,8,12)&6&(3,5,12,7,10,16)(8,15,9)&6&(3,5,8,15,6,12,9,11)&6\\
\hline
(3,6,11,14,9,8,12)(5,15)&6&(3,7,12,5,4)(6,15,14)&6&(3,7,5,4,12,13)(6,14)&6\\
\hline
(3,9,10,5,6,15,11)(7,12)&6&(3,9,12,7,10,11)(4,15,5)&6&(3,9,8,15,7,10,11)(5,12)&6\\
\hline
(3,10,11,6,12,5)(7,15)&6&(3,11,5,15,9,8,12,7,10)&6&(3,11,7,12,14,6,15,5)&6\\
\hline
(3,11,7,12,6,9,15,10,5,13)&6&(3,11,7,15,5,4)(6,12)&6&(3,11,7,15,6,12,5,4)&6\\
\hline
(3,11,7,9,15,6,12,10,5,13)&6&(4,11,9)(5,15,6,12)&6&(5,6,12)(7,15,9,10,11)&6\\
\hline
(5,8,15,10,6,12,9,11)&6&(5,10)(6,16,7,9,11,8,14)&6&(5,15)(7,8,12,9,10,11)&6\\
\hline
(5,15,7,8,12,9,10,11)&6&(1,16,7)(2,15,5,4,12,14,6)&12&(3,9,12,14,6,11)(4,15,5)&12\\
\hline
(5,15,6,11,9,8,12)&12&(2,6,12,7,16,13)(3,5)&24&(3,6,7,12,9,11,10)(5,15)&24\\
\hline
(2,5,3,6,12,7,16,13)&48&&&&\\
\hline
\end{tabular}
\end{table}

\begin{table}[htb]
\centering \caption{Generating permutations and $|\Aut(C)|$ for
codes with $C_{\varphi}= 2f_8$} 
\label{Table:4} \small
\begin{tabular}{|l|c||l|c||l|c|}
\hline
\text{permutation}& $|\Aut|$&   \text{permutation}& $|\Aut|$&   \text{permutation}& $|\Aut|$\\
\hline
(1,11,12,5,3,13,2,8,15,9)&3&(1,12,3,7,10)(4,5,14,13,16)&3&(2,3,11,12,8,7,5)&3\\
\hline
(2,5,11,12,3,6)(9,10)&3&(2,5,11,6)(8,14,12,15)&3&(2,5,3,10,6,4,11)(7,15)&3\\
\hline
(2,5,3,6)(8,15,16,9,10)&3&(2,5,4,12,3,6)(9,16,10)&3&(2,5,4,3,6)(9,15,16,10)&3\\
\hline
(2,5,7,3,6,15,9,11)&3&(2,5,8,11,12,3,6)(9,10)&3&(2,5,8,15,10,11,6)&3\\
\hline
(2,5,8,15,11,3,6)&3&(2,6)(3,5)(8,15,16,9,10)&3&(2,6)(3,5,4)(9,11,12)&3\\
\hline
(2,6)(3,5,4,15,16,10,9)&3&(2,6)(3,5,4,16,14,9,12)&3&(2,6)(3,5,7)(9,15,10,16)&3\\
\hline
(2,6)(3,5,8,4,10,11,12)&3&(2,6)(3,8,16,9,10,5,12)&3&(2,8,12,3,6)(5,16,9,10)&3\\
\hline
(2,8,4,10,5,11,12,3,6)&3&(2,8,4,12,3,6)(5,11)&3&(2,9,10,11,5)(3,8,7,12)&3\\
\hline
(2,9,11,5,15,8,3,6)&3&(2,9,11,5,3,6)(8,15)&3&(2,9,11,6,8,16,5,12)&3\\
\hline
(2,9,11,8,3,5,15,16,14,6)&3&(2,9,15,5,3,8,11,12,14,6)&3&(2,9,8,14,5,10,6)(15,16)&3\\
\hline
(2,9,8,16,5,12,11,6)&3&(2,9,8,5,6)(12,14,15)&3&(2,9,8,5,6)(12,16,15)&3\\
\hline
(2,9,8,5,6)(15,16)&3&(2,9,8,6)(12,16,15)&3&(2,10,7,9,15)(3,6,11,5)&3\\
\hline
(2,11,12,8,7,5)&3&(2,11,14,15,9,4,5,7)&3&(2,11,3,5,8,4)(6,12)&3\\
\hline
(2,11,6)(5,12,9,8,16)&3&(2,12,16,10,3,9,7,5,4)&3&(2,15)(5,7,10)(6,14)(9,11)&3\\
\hline
(2,15,16,10,9,5,4,6)&3&(2,15,16,11,9,10,5,6)&3&(2,15,16,14,10,8,6,7)&3\\
\hline
(2,15,16,14,6)(4,10,5)&3&(2,15,16,14,6,12,5,4)&3&(2,15,16,14,8,6,7)&3\\
\hline
(2,15,16,9,10,5,14,6)&3&(2,15,16,9,10,5,14,6,7)&3&(2,15,16,9,10,5,6,7)&3\\
\hline
(2,15,5,12,6)(8,16)(9,10)&3&(2,15,5,14,6,7)(9,11,12)&3&(2,15,5,16,8,12,9,11,13,4)&3\\
\hline
(2,15,5,8,10,6)(9,11,12)&3&(2,15,6)(4,11,9)(5,16)&3&(2,15,6)(5,10)(8,14,9,11)&3\\
\hline
(2,15,6)(5,10,9,11,8,14)&3&(2,15,6)(5,11)(8,14)(9,10)&3&(2,15,6)(5,12)(8,16)(9,10)&3\\
\hline
(2,15,8,10,13,4)(5,14,16)&3&(2,15,8,6)(5,12,9,10,11)&3&(3,5)(4,11,7,15,16,6)&3\\
\hline
(3,5,12,6,11,9,8,16)&3&(3,6,4,11,14,5)(7,15,10)&3&(3,7,9,15,6,11,14,12,5)&3\\
\hline
(3,12,5,7)(6,15)(9,11)&3&(3,15,6,8,10,5,14,9,11)&3&(3,16,14,5)(6,11,7,9,15)&3\\
\hline
(3,16,5)(6,8,15,9,11)&3&(3,16,9,11,5,7,15,6)&3&(4,7,11,12,6)(5,15)&3\\
\hline
(5,7,12,9,10)&3&(5,7,14)(6,10)(9,11)&3&(5,15)(6,9,11,12)&3\\
\hline
(1,2,14,15,16,8,11,13,9)&6&(1,2,15,16,13,7,5,10,14,9)&6&(1,2,9)(5,15,11,12,13,7)&6\\
\hline
(1,15,13,12,9)(5,7,10,11)&6&(1,15,9)(5,6,12)(10,13)&6&(2,5,3,6)(8,15)(9,11,10)&6\\
\hline
(2,5,6,10,16)(4,12,13)(9,15)&6&(2,6)(3,14,5,9,15,16)&6&(2,6)(5,10,8,16)(9,11)&6\\
\hline
(2,6,11,5,16,15,8)(9,10)&6&(2,6,7,11)(5,12)(9,10)&6&(2,6,7,8,9,11,12,14)&6\\
\hline
(2,6,8,11,13,4)(3,5)&6&(2,7,11,12)(5,9,10)(6,14)&6&(2,7,11,12)(5,9,10,6,14)&6\\
\hline
(2,8,10,3,13,4)(5,15,16)&6&(2,8,10,5,15,16,3,13,4)&6&(2,9,11,5,3,8,15,14,6)&6\\
\hline
(2,11,12,13,4)(5,10)(8,15)&6&(2,11,12,3,13,4)(5,10,9,7)&6&(2,12,10,16,6,7,8,4)&6\\
\hline
(2,12,16,9,4,10,6)(8,14)&6&(2,15,11,12,10,9,8,5,6)&6&(2,15,11,12,14,6,8,13,4)&6\\
\hline
(2,15,16,10,13,4,11,8,6)&6&(2,15,16,14,6,8,4)&6&(2,15,16,6)(4,10,8,14,9)&6\\
\hline
(2,15,5,11,6)(4,16,8)&6&(2,15,5,9,10,11,6)&6&(2,15,9,4,5,7,10,11,12)&6\\
\hline
(3,5,9,11)(6,10,12)(7,14)&6&(3,9,5,7,10,11,12)&6&(3,16)(5,8,9,15)(6,12)&6\\
\hline
(3,16)(5,9,15)(6,12)&6&(3,16)(5,9,15,6,12)&6&(3,16,15,6,7,9,11)&6\\
\hline
(5,9,10,6,11,12,7,14)&6&(5,9,10,6,14)(7,11,12)&6&(5,9,11,12,7,14)(6,10)&6\\
\hline
(1,16,7,14,6,11,8,2,4,3)&12&(1,15,3,9)(5,7,10,11,13,12)&24&(2,5,7,10,11,13,4,12,9,15)&24\\
\hline
(4,15,5,7,10,16,13)(9,12)&24&&&&\\
\hline
\end{tabular}
\end{table}

\end{document}